\documentclass[
  aps,
  prl,
  reprint,
  longbibliography,
  nofootinbib
]{revtex4-2}

\usepackage[T1]{fontenc}
\usepackage[english]{babel}
\usepackage{graphicx}
\usepackage{xcolor}
\usepackage{rotating}

\usepackage{amsmath, amssymb, graphics, amsthm, isomath}
\usepackage{physics}
\usepackage{verbatim}
\usepackage{tikz}
\usepackage{tikz-cd}
\usepackage{circuitikz}
\usetikzlibrary{patterns}
\usetikzlibrary{through,hobby}
\usetikzlibrary{decorations.markings}
\usetikzlibrary{fadings,shadings}
\usetikzlibrary{plotmarks}
\usetikzlibrary{positioning}
\usetikzlibrary{decorations,arrows}
\usetikzlibrary{decorations.pathreplacing}
\usetikzlibrary{chains, scopes, positioning, backgrounds, shapes, fit, shadows, calc, arrows.meta, decorations.pathreplacing}
\usepackage{blkarray}
\usepackage{svg}
\usepackage{mathrsfs}
\usepackage{algorithm}
\usepackage{algpseudocode}

\usepackage{xcolor}
\definecolor{mycitecolor}{rgb}{0.0, 0.45, 0.85}   

\newcommand{\im}{\operatorname{im}}

\usepackage[
  colorlinks=true,
  linkcolor=mycitecolor,
  citecolor=mycitecolor,
  urlcolor=mycitecolor,
  hyperindex=true,
  linktocpage=true
]{hyperref}

\usepackage[capitalise,compress]{cleveref}
\usepackage{tcolorbox}

\newtcolorbox{shadedtheorem}{
  colback=gray!15,    
  colframe=white,     
  boxrule=0pt,        
  arc=0pt,            
  left=5pt,           
  right=5pt,
  top=5pt,
  bottom=5pt
}

\renewcommand\epsilon{\varepsilon}

\theoremstyle{upright-hang}
\newtheorem{thm}{Theorem}
\newtheorem{cor}[thm]{Corollary}

\newtheorem{lem}[thm]{Lemma}

\usepackage{mathtools}
\tikzstyle{densely dashed}= [dash pattern=on 4pt off 3pt]

\usepackage{dsfont}

\allowdisplaybreaks

\usepackage{centernot}

\begin{document}

\definecolor{darkred}{RGB}{140,0,0} 

\makeatletter
\def\frontmatter@above@affiliation{\addvspace{5pt}}
\makeatother

\preprint{IQuS@UW-21-138}

\title{Self-Correcting 3D Surface Code with Growing Local Dimension}

\author{Jinkang Guo}
\affiliation{InQubator for Quantum Simulation, Department of Physics, University of Washington, Seattle, WA 98195, USA}

\date{\today}

\begin{abstract}

We construct self-correcting 3D surface codes over \(\mathbb Z_q\) with local dimension \(q\) growing with system size. The model preserves the local geometry and stabilizer support of the ordinary 3D surface code and remains translation invariant in the bulk, while modifying only the stabilizer Pauli powers. We use the Lee metric to characterize error weight and consider local Metropolis dynamics with single-qudit \(X^{\pm1}\) and \(Z^{\pm1}\) updates, each changing the energy by at most a constant independent of $q$ and system size $N$. We prove that, below the critical temperature, when the local dimension scales as \(q=\Theta(N)\), the memory times in both error sectors grow at least as \(\exp\left[\Omega(N^{1/3})\right]\). Our numerical results further support these analytical results. Together, these results establish growing local Hilbert-space dimension as a new resource for designing self-correcting quantum memories while retaining simple local geometry and local thermal dynamics.

\end{abstract}

\maketitle


\section{Introduction} \label{sec:intro}

Self-correcting quantum memories aim to preserve quantum information without continual active error correction \cite{Terhal2015quantum,Brown2016quantum,Hong2025quantum}. The 4D toric code provides the canonical example \cite{Dennis2002topological,Alicki2010thermal}. Throughout this work, we use self-correction in the sense of Ref.~\cite{Balasubramanian2026passive}: under local thermalization dynamics, the memory lifetime grows at least as \(\exp(N^\gamma)\) for some \(\gamma>0\) below a nonzero critical temperature. Achieving passive self-correction in three dimensions has long been a central challenge \cite{Brown2016quantum,Alicki2009thermalization,BravyiTerhal2009nogo,Yoshida2011feasibility,LandonCardinal2013local,Haah2011local,Kim2012three,BravyiHaah2013quantum,Stengele2026haah,Michnicki2014three,SivaYoshida2017marginal,Roberts2020symmetry,Williamson2024layer,Gu2026layer,Williamson2026partial}. Recently, a self-correcting quantum memory in three dimensions under local thermal dynamics was proposed in Ref.~\cite{Balasubramanian2026passive}. Their result also illustrates how difficult this problem is. The construction of the code is recursive, and establishing its self-correcting behavior requires a highly nontrivial proof. 
Beyond passive self-correction, long-lived low-dimensional quantum memories have also attracted considerable recent interest \cite{Tsao2026entropic}. Existing approaches include local or autonomous decoding
\cite{Herold2015cellular,Herold2017cellular,Kubica2019cellular,Balasubramanian2026local,Dunnweber2026quantum,Shtanko2025bounds}
and engineered dissipation
\cite{Pastawski2011engineered,Fujii2014measurement,Lieu2024candidate,Lieu2020symmetry,Liu2024dissipative}, which introduce additional correction processes during storage, as well as long-range interactions mediated by auxiliary systems \cite{Hamma2009toricboson,Pedrocchi2013enhanced}, which modify the interaction structure. Here we explore a different resource tradeoff: we retain the simple local geometry of the 3D surface code~\cite{Kitaev2003fault,Hamma2005string}, translational invariance in the bulk and local thermal dynamics, while allowing the local Hilbert-space dimension \(q\) to scale with system size. The growth of \(q\) places the construction outside the scope of the no-go theorem of Ref.~\cite{Yoshida2011feasibility}, which assumes a fixed local Hilbert-space dimension.

Large-\(q\) qudits also have independent physical motivation \cite{Brock2025quantum}. Many multi-level quantum systems naturally provide local Hilbert-space dimensions larger than two, while in nuclear-spin systems \cite{Asaad2020coherent} and continuous-variable bosonic encodings such as GKP \cite{Gottesman2001encoding} or cat codes \cite{Cochrane1999macroscopically,Albert2018performance}, noise can favor small Pauli shifts over large ones \cite{Guo2026quantumcodesleemetric}. For such noise, the Hamming weight records only whether a qudit is affected and therefore discards the shift magnitude, whereas the Lee metric measures the cyclic displacement in \(\mathbb Z_q\) and distinguishes small from large shifts \cite{Lee1958properties}. Accordingly, the local Metropolis dynamics considered below takes a single-qudit \(X^{\pm1}\) or \(Z^{\pm1}\) as the elementary error, so one step changes the corresponding Pauli exponent only by one modulo \(q\),  rather than treating all nonzero single-qudit shifts on the same footing, as in standard Hamming-type qudit error models \cite{Anwar2014fast}. Quantum codes in this Lee-metric setting were studied in Ref.~\cite{Guo2026quantumcodesleemetric}, where increasing \(q\) was shown to qualitatively change thermal memory under local Metropolis dynamics, providing direct motivation for treating the local Hilbert-space dimension as a resource for thermal protection.

Our construction starts from the ordinary 3D \(\mathbb Z_q\) surface code, whose two logical operators have qualitatively different geometries \cite{Bullock2007qudit,Aloshious2019local,Lee2025ZNgeneralizations}. One logical operator is membrane-like and has a growing energy barrier under local errors, whereas the conjugate logical operator is string-like and remains thermally unstable \cite{Castelnovo2008topological,Nussinov2008autocorrelations}. Our modification directly targets this string-like logical operator. We keep the support of every stabilizer unchanged and leave the Pauli powers in the other two spatial directions exactly as in the ordinary 3D surface code, modifying only the powers along the direction of the string-like logical operator while keeping the code translationally invariant in the bulk. This modification makes the Pauli powers of the string-like logical operator follow the sequence \(1,2,4,8,\ldots\) modulo \(q\). We refer to the resulting code as the weighted 3D surface code. This modification does not change the error support in the Hamming metric, but in the Lee metric, the growing local Hilbert-space dimension can directly change the thermal stability of the string-like sector.

Our main result is that the weighted 3D surface code is self-correcting at sufficiently low temperature with growing $q$. For \(N\) physical qudits and odd \(q=\Theta(N)\), we prove that the memory times of both error sectors scale as \(\exp[\Omega(N^{1/3})]\). Measured by Hilbert-space dimension, \(N\) \(q\)-level qudits with \(q=\Theta(N)\) correspond to \(\Theta(N\log N)\) qubits, so the growing local dimension introduces only a logarithmic multiplicative overhead in this counting. Our numerical results further suggest a more general scaling \(\log t_{\mathrm{mem}}\sim \min\{q,N^{1/3}\}\) when $q \gg \log N$. In particular, for \(q=N^\gamma\) with any fixed \(\gamma>0\), \(t_{\mathrm{mem}}\) scales as $\exp[\Theta(N^{\min(\gamma,1/3)})]$.

\section{The construction}

We first define the following chain complex over \(\mathbb Z\), and then reduce its boundary maps modulo \(q\) to obtain a CSS code over \(\mathbb Z_q\) \cite{Calderbank1996good,Steane1996error,Hostens2005stabilizer}, where \(q\) is odd and \(L\) is the linear system size. The construction uses two ordinary
open-boundary repetition-code complexes in the $x$ and $y$ directions and one helical repetition-code complex, introduced in Ref.~\cite{Guo2026quantumcodesleemetric}, in the \(z\) direction,
\begin{equation} \label{eq:chaincomplex1}
\begin{aligned}
A_\bullet^x &: \quad 0 \longrightarrow \mathbb{Z}^{L-1}
\xrightarrow{R_L^T} \mathbb{Z}^L \longrightarrow 0, \\
A_\bullet^y &: \quad 0 \longrightarrow \mathbb{Z}^{L-1}
\xrightarrow{R_L^T} \mathbb{Z}^L \longrightarrow 0, \\
B_\bullet^z &: \quad 0 \longrightarrow \mathbb{Z}^L
\xrightarrow{K_L} \mathbb{Z}^{L-1} \longrightarrow 0,
\end{aligned}
\end{equation}
where
\begin{equation}
(R_L)_{j,k}=\delta_{j,k}-\delta_{j+1,k},
\qquad
(K_L)_{j,k}=-2\delta_{j,k}+\delta_{j+1,k}
\end{equation}
with $1\leq j\leq L-1$ and $1\leq k\leq L$.

The three-dimensional chain complex is their tensor product \cite{Tillich2014quantum,Zeng2019higher,Bravyi2014homological},
\begin{align}
&C_\bullet=A_\bullet^x\otimes A_\bullet^y\otimes B_\bullet^z, \notag
\\
&0\longrightarrow C_3\xrightarrow{\partial_3}C_2
\xrightarrow{\partial_2}C_1\xrightarrow{\partial_1}C_0\longrightarrow 0.
\end{align}
Physical qudits are associated with the basis elements of $C_1\otimes\mathbb{Z}_q$. After reducing the boundary
maps modulo $q$, the CSS parity-check matrices are
\begin{equation}
H_X=\partial_1 \pmod q,
\qquad
H_Z=\partial_2^T \pmod q.
\end{equation}
The stabilizers commute because $H_XH_Z^T=\partial_1\partial_2=0$. The bulk stabilizer types are shown in Fig.~\ref{fig:bulk-stabilizers}.

For \(a\in\{0,\ldots,q-1\}\), define its Lee weight as $|a|_L=\min\{a,q-a\}$ \cite{Lee1958properties}, and extend the Lee weight additively to vectors. The row and column Lee weights of both \(H_X\) and \(H_Z\) are bounded by constants independent of \(L\) and \(q\), so the code is Lee-LDPC \cite{Guo2026quantumcodesleemetric}.

\begin{figure}[t]
\centering
\begingroup
\begin{tikzpicture}[
    x={(1cm,0cm)},
    y={(0.70cm,0.45cm)},
    z={(0cm,1cm)},
    line cap=round,
    line join=round,
    edge/.style={draw=black, line width=0.8pt},
    face/.style={edge, fill=black!3},
    weight/.style={
        font=\footnotesize, fill=none,
        inner xsep=1.2pt, inner ysep=0.6pt
    },
    wbottom/.style={below=2pt},
    wright/.style={left=3pt},
    wtop/.style={below=3pt},
    wleft/.style={right=3pt},
    heading/.style={font=\small},
    plane/.style={font=\footnotesize}
]

\def\arm{1.50}
\def\L{1.40}
\def\gap{0.40}

\def\drawface#1#2#3#4{%
    \filldraw[face] (A)--(B)--(C)--(D)--cycle;
    \foreach \P in {A,B,C,D}
        \fill (\P) circle[radius=1pt];
    \path
        (A)--node[weight,wbottom] {$#1$} (B)
           --node[weight,wright] {$#2$} (C)
           --node[weight,wtop] {$#3$} (D)
           --node[weight,wleft] {$#4$} (A);
}

\coordinate (O)  at (0,0,0);
\coordinate (Xp) at (\arm,0,0);
\coordinate (Xm) at (-\arm,0,0);
\coordinate (Yp) at (0,\arm,0);
\coordinate (Ym) at (0,-\arm,0);
\coordinate (Zp) at (0,0,\arm);
\coordinate (Zm) at (0,0,-\arm);

\foreach \P in {Xp,Xm,Yp,Ym,Zp,Zm}{
    \draw[edge] (O)--(\P);
    \fill (\P) circle[radius=1pt];
}
\fill (O) circle[radius=1.3pt];

\path (O)--node[pos=0.58,weight,below=3pt] {$1$} (Xp);
\path (O)--node[pos=0.58,weight,above=3pt] {$-1$} (Xm);
\path (O)--node[pos=0.72,weight,above=3pt] {$1$} (Yp);
\path (O)--node[pos=0.72,weight,below=3pt] {$-1$} (Ym);
\path (O)--node[pos=0.58,weight,left=3pt] {$1$} (Zp);
\path (O)--node[pos=0.58,weight,right=3pt] {$-2$} (Zm);

\node[plane, anchor=west,  xshift=3pt] at (Xp) {$x$};
\node[plane, anchor=west,  xshift=3pt] at (Yp) {$y$};
\node[plane, anchor=south, yshift=3pt] at (Zp) {$z$};

\node[heading] at (0,0,2.30) {$X$ stabilizer};
\node[heading] at (4.29,0,2.30) {$Z$ stabilizers};

\begin{scope}[shift={(2.90,0,-1.215)}]

    \begin{scope}
        \coordinate (A) at (0,0,0);
        \coordinate (B) at (\L,0,0);
        \coordinate (C) at (\L,0,\L);
        \coordinate (D) at (0,0,\L);
        \drawface{-1}{-1}{+2}{1}
        \path (A)--node[midway, below=18pt, plane] {$xz$} (B);
    \end{scope}

    \begin{scope}[
        yshift=\gap cm,
        wright/.style={right=3pt},
        wtop/.style={above=2pt},
        wleft/.style={left=3pt}
    ]
        \coordinate (A) at (0,0,\L);
        \coordinate (B) at (\L+0.23,0,\L);
        \coordinate (C) at (\L+0.23,0.92*\L,\L);
        \coordinate (D) at (0,0.92*\L,\L);
        \drawface{1}{1}{-1}{-1}
        \node[plane,anchor=south east,xshift=-3pt,yshift=3pt]
            at (D) {$xy$};
    \end{scope}

    \begin{scope}[xshift=\gap cm,wright/.style={right=3pt}]
        \coordinate (A) at (\L,0,0);
        \coordinate (B) at (\L,\L,0);
        \coordinate (C) at (\L,\L,\L);
        \coordinate (D) at (\L,0,\L);
        \drawface{-1}{-1}{+2}{1}
        \path (A)--node[midway, below=18pt, plane] {$yz$} (B);
    \end{scope}

\end{scope}
\end{tikzpicture}
\endgroup
\caption{Bulk stabilizers of the weighted 3D surface code.
The vertex-centered \(X\) stabilizer has coefficients
\((+x,-x,+y,-y,+z,-z)=(1,-1,1,-1,1,-2)\).
For $Z$ stabilizers, the \(xy\) plaquette has the oriented edge pattern
\((1,1,-1,-1)\).
For each vertical \(xz\) or \(yz\) plaquette, the bottom, right,
top and left edges carry \((-1,-1,+2,1)\), respectively.
All panels use the same coordinate projection.}
\label{fig:bulk-stabilizers}
\end{figure}
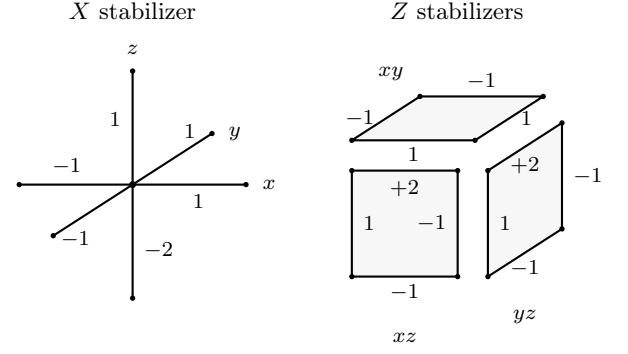

The only nonzero homology groups of the three one-dimensional factors are
\begin{equation}
H_0(A_\bullet^x;\mathbb{Z})
\cong
H_0(A_\bullet^y;\mathbb{Z})
\cong
\mathbb{Z},
\qquad
H_1(B_\bullet^z;\mathbb{Z})
\cong
\mathbb{Z},
\end{equation}
where $H_1(B_\bullet^z;\mathbb{Z})$ is generated by the vector
\begin{equation}
u_L
=
\left(
1,2,2^2,\ldots,2^{L-1}
\right)^T.
\end{equation}
The K\"unneth formula~\cite{Hatcher2002algebraic} therefore gives
\begin{equation}
H_1(C_\bullet;\mathbb{Z})
\cong
H_0(A_\bullet^x;\mathbb{Z})
\otimes
H_0(A_\bullet^y;\mathbb{Z})
\otimes
H_1(B_\bullet^z;\mathbb{Z})
\cong
\mathbb{Z},
\end{equation}
with no torsion. By the universal coefficient theorem~\cite{Hatcher2002algebraic}, after tensoring the integral complex with $\mathbb{Z}_q$,
\begin{equation}
H_1\left(
C_\bullet\otimes_{\mathbb{Z}}\mathbb{Z}_q
\right)
\cong
\mathbb{Z}_q,
\end{equation}
so the code encodes a single logical qudit.

The corresponding logical operators are obtained by reducing the following integer representatives
modulo $q$. For any vertical column, label its $L$ $z$-oriented edges by $e_1,\ldots,e_L$ in increasing $z$ order.
Let $\Sigma_{xy}$ be the $xy$ plane crossing $e_1$, and let $E_z(\Sigma_{xy})$ denote the set of all $z$-oriented edges crossing
this plane. A pair of logical generators can be represented by
\begin{subequations}
\begin{align}
\ell_Z=\sum_{k=1}^{L}2^{k-1}e_k,
\end{align}
\begin{align}
\ell_X=\sum_{e\in E_z(\Sigma_{xy})}e.
\end{align}
\end{subequations}
They overlap only on $e_1$, and therefore $\langle \ell_X,\ell_Z\rangle=1$, so their reductions modulo $q$ are the generators of all logical operators.

\section{Main result}

We consider the discrete-time single-site Metropolis dynamics \cite{Metropolis1953equation} separately in
the two error sectors, with energies
\begin{equation}
E_Z(z)=|H_X z|_L,
\qquad
E_X(x)=|H_Z x|_L.
\end{equation}
At each step, one physical qudit is selected and its error coordinate is proposed to change by $+1$ or $-1$, so the energy change of every elementary update is bounded by a constant independent of \(L\) and \(q\). A move with energy change $\Delta E$ is accepted with probability
\begin{equation}
A(\Delta E)=\min\{1,e^{-\beta\Delta E}\}.
\end{equation}
The dynamics is reversible with respect to the corresponding Gibbs distribution.
The discrete-time formulation is used only for convenience; the same single-site updates can be implemented as a continuous-time local Metropolis Lindbladian, with only a polynomial rescaling of time that does not affect the exponential memory-time scaling.

We define the \(Z\)- and \(X\)-sector memory times as follows. For each sector, we first construct a bottleneck separating the regions associated with distinct codewords, and then define the decoder by assigning each region to its codeword. Decoding can fail only after the trajectory has crossed the corresponding bottleneck. The \(Z\)-sector dynamics starts from an arbitrary codeword, whereas the \(X\)-sector dynamics starts from the thermal encoding of the corresponding codeword. The memory time is the first time at which this failure probability reaches \(1/4\).

We prove that there exist constants \(c,c_Z,c_X>0\) and \(\beta_c<\infty\), independent of \(L\) and \(q\), such that for every \(\beta>\beta_c\), all sufficiently large \(L\), and odd \(q\) satisfying
\begin{equation}
q\geq c L^3,
\end{equation}
the memory times of both error sectors satisfy
\begin{equation}
t_{\mathrm{mem}}^Z \geq e^{c_Z L},
\qquad
t_{\mathrm{mem}}^X \geq e^{c_X L}.
\end{equation}
Since the number of physical qudits is $N=\Theta(L^3)$,
both times scale as $e^{\Omega(N^{1/3})}$. Thus, under the single-site Metropolis dynamics considered here,
the weighted 3D surface code is a self-correcting quantum memory.

\noindent\textit{Proof sketch}

The rigorous proofs are given in Appendices \ref{app:Z} and \ref{app:X}; here we only provide a brief overview of the main ideas.

The \(X\)-error sector is essentially the same as in the ordinary 3D surface code. If every connected \(Z\)-syndrome cluster is sufficiently small, it can be corrected locally, so a change of the logical \(X\) sector must pass through a configuration containing a connected cluster of size \(\Omega(L)\). Up to polynomial factors in \(L\), the number of connected syndrome supports at this scale is bounded by $e^{cL}$ for some constant \(c>0\).
On the other hand, for any fixed connected support of size \(\alpha L\), summing over all nonzero \(\mathbb Z_q\) syndrome values gives a total thermal weight bounded by
\begin{equation} \label{eq:boltzmanweight}
\left(\frac{2}{e^\beta-1}\right)^{\alpha L}
=
\exp\!\left[
-\alpha L\log\!\left(\frac{e^\beta-1}{2}\right)
\right].  
\end{equation}
The exponent of $e^{cL}$ is fixed by the local geometry, whereas the decay rate in (\ref{eq:boltzmanweight}) increases with \(\beta\). Thus, at sufficiently low temperature, the Boltzmann suppression dominates the configurational entropy, with a critical inverse temperature of order one independent of \(q\). The same competition increasingly suppresses larger clusters, so configurations containing a connected cluster of size at least \(\alpha L\) have total Gibbs weight \(e^{-\Omega(L)}\), giving an exponentially long \(X\)-sector memory time.

For the \(Z\)-error sector, it is useful to first ignore reduction modulo \(q\). Along a vertical column, the logical \(Z\) operator has error exponents \(1,2,4,\ldots,2^{L-1}\). The Lee distance of the $Z$ sector is therefore exponentially large and reaching a different logical sector requires an exponential number of single-site updates.

For finite \(q\), this exponential growth in the Lee distance is cut off by reduction modulo \(q\): an alternative way to escape the logical sector is for an endpoint syndrome to approach the maximal Lee weight, of order \(q/2\). The rigorous proof makes these two escape routes precise by lifting the dynamics to \(\mathbb Z\) and surrounding each codeword by a cylinder. Reaching a cap of the cylinder requires at least \(2^{L-2}\) single-site updates, while reaching its side requires a syndrome coordinate with Lee weight of order \(q\). For \(q\ge cL^3\) and sufficiently low temperature, the total Boltzmann weight near the side is exponentially small, so either escape route requires an exponentially long time. After reducing modulo \(q\), the cylinders remain pairwise disjoint, and each cylinder is mapped injectively, so the dynamics also agrees with the lifted dynamics until the cylinder is exited, giving an exponentially long \(Z\)-sector memory time.

\section{Numerical results}

We first present numerical results supporting our analytical conclusions in the scaling regime covered by the rigorous bounds.
For each sector, we initialize the dynamics at the zero-error codeword and
record the first-passage time \(\tau_A\) to the corresponding bottleneck,
reporting
\begin{equation}
t_{\mathrm{mem}}^A(1/4)
=
\inf\{t:\Pr(\tau_A\leq t)\geq 1/4\}.
\end{equation}

For the \(Z\) sector, the bottleneck is the boundary of the cylinder: either the longitudinal coordinate \(T\) along \(\ell_Z\) (Appendix~\ref{app:Z}) exceeds half its logical period, or a syndrome coordinate reaches Lee weight \((q-3)/2\).
For the \(X\) sector, it is reached when the largest connected component of
the nonzero \(Z\) syndrome has size at least \(\lceil L/2\rceil\).
Right-censored trajectories are included using the Kaplan--Meier estimator \cite{Kaplan1958nonparametric}.

As shown in Fig.~\ref{fig:xz_memory}, at $\beta=2$ and $q\simeq L^3$, the memory times in both the $X$- and $Z$-error sectors grow exponentially with the linear system size $L$. For the system sizes included in the fits, we obtain
\begin{equation}
t_{\mathrm{mem}}^X\propto e^{0.360L},
\qquad
t_{\mathrm{mem}}^Z\propto e^{1.510L}.
\label{eq:numerical-memory-time-L}
\end{equation}
The two $X$-sector data points at the smallest $L$ are excluded from the fit. These numerical results are consistent with the main result.

\begin{figure}[t]
    \centering
    \includegraphics[width=0.97\linewidth]{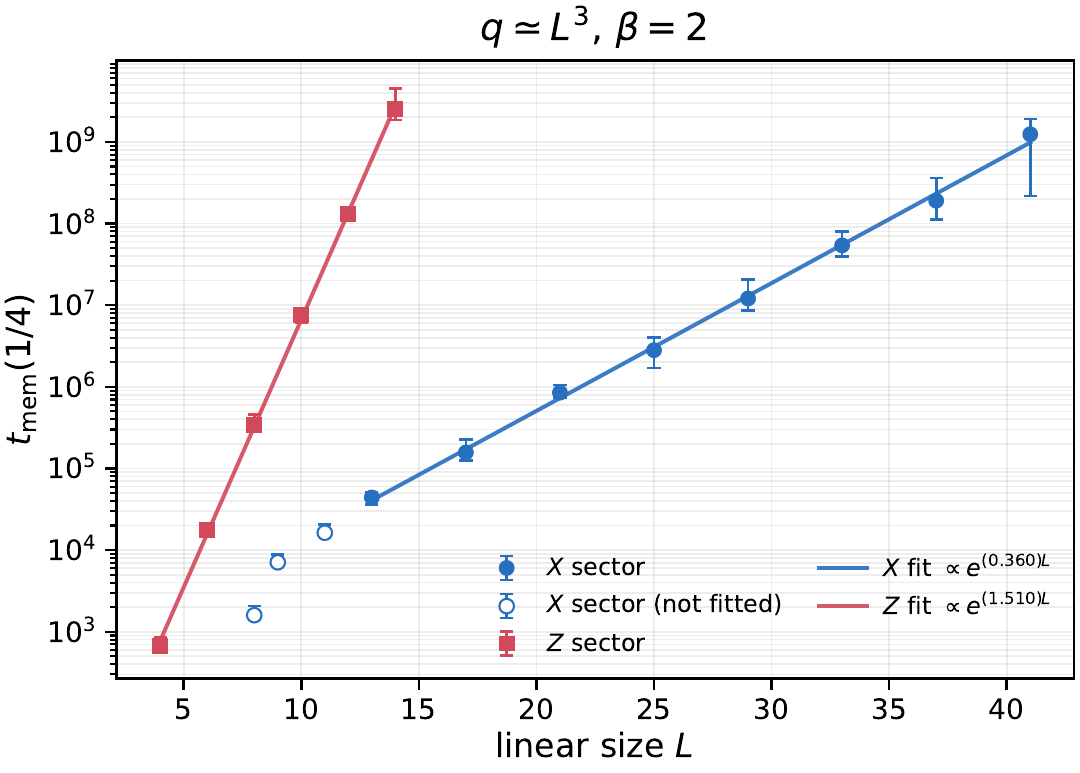}
    \caption{Memory times of the \(X\)- and \(Z\)-error sectors at \(\beta=2\) with \(q\simeq L^3\). Both sectors exhibit exponential growth with the linear system size \(L\). The solid lines show fits \(t_{\mathrm{mem}}^X(1/4)\propto e^{0.360L}\) and \(t_{\mathrm{mem}}^Z(1/4)\propto e^{1.510L}\); the open \(X\)-sector data points at small \(L\) are not included in the fit.}
    \label{fig:xz_memory}
\end{figure}

We further examine the dependence of the $Z$-sector memory time on the local dimension $q$. At $\beta=1$, we consider several families with $q\simeq a_\alpha L^\alpha$, where $\alpha=0.25,0.5,0.75,1$. As shown in Fig.~\ref{fig:z_power}, when the memory time is plotted directly as a function of $q$, the data for different $\alpha$ collapse onto a single curve and are well described by the global fit
\begin{equation}
t_{\mathrm{mem}}^Z\propto e^{0.477q}.
\label{eq:numerical-memory-time-q}
\end{equation}
This collapse suggests that, in this regime, the $Z$-sector memory time is controlled primarily by $q$ rather than by $L$.

In the proof in Appendix \ref{app:Z}, we used a very loose upper bound to control the entropic contribution in the $Z$ sector, which is also the main source of the stringent requirement on $q$ in the analytical result. The numerical results suggest that a much tighter bound should be possible. This behavior has a simple heuristic interpretation. As discussed in Appendix~\ref{app:Z-q}, the extensive entropic factor entering the rigorous bound may largely cancel after normalization, leaving only a polynomial number of choices for which syndrome coordinate approaches the maximal Lee weight. The corresponding entropy is therefore only \(O(\log N)\), whereas reaching this region costs an energy of order \(q\). This suggests that the side-escape mechanism remains suppressed already when \(q\gg\log N\), consistent with the numerical scaling in (\ref{eq:numerical-memory-time-q}).
Combining this with the \(e^{\Omega(L)}\) cap time suggests
\begin{equation}
\log t_{\mathrm{mem}}^Z\sim \min\{q,N^{1/3}\}.
\end{equation}
In particular, any polynomial growth \(q=N^\gamma\) would give a stretched-exponential memory time.

\begin{figure}[t]
    \centering
    \includegraphics[width=0.97\linewidth]{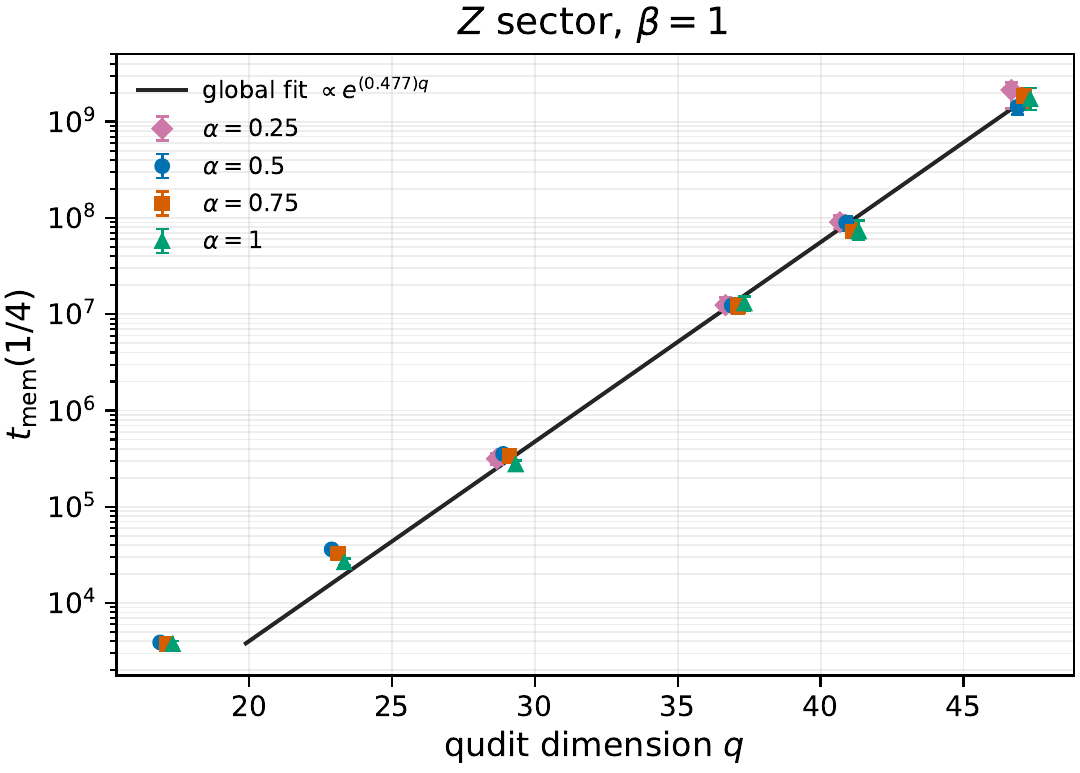}
    \caption{\(Z\)-sector memory time at \(\beta=1\) for \(q\simeq a_\alpha L^\alpha\), with \(\alpha=0.25,0.5,0.75,1\) and \(a_{0.25}=14\), \(a_{0.5}=6\), \(a_{0.75}=3\), and \(a_1=2\), plotted as a function of the qudit dimension \(q\). Despite the different values of \(\alpha\) and \(a_\alpha\), the data collapse onto a common curve, with a global fit \(t_{\mathrm{mem}}(1/4)\propto e^{0.477q}\), suggesting that the memory time is controlled by \(q\) in this regime.}
    \label{fig:z_power}
\end{figure}

\section{Discussion and Outlook}

Self-correcting quantum memories are usually studied at fixed local Hilbert-space dimension, while the code geometry, interaction structure, or dynamics is varied \cite{Terhal2015quantum,Brown2016quantum}. Our results demonstrate another possibility: allowing \(q\) to grow with system size can itself serve as a resource for enhancing the thermal stability of low-dimensional quantum memories. In our construction, this resource yields passive thermal protection in three dimensions without long-range interactions or continual active error correction. The code retains the local geometry and stabilizer support of the ordinary 3D surface code and is translation invariant in the bulk.

A natural question is whether growing local dimension can provide a similar route to self-correction in two dimensions. However, for the single-qudit \(X^{\pm1}\) and \(Z^{\pm1}\) shift errors considered here, the transference bounds of Ref.~\cite{Guo2026quantumcodesleemetric} imply that, for any quantum CSS Lee-LDPC code, at least one of the two error sectors has Lee distance \(O(N)\), uniformly in \(q\). This obstructs a direct two-dimensional extension of the mechanism used here. One possible direction is to ask whether growing local dimension can enable a two-dimensional memory stabilized by active quantum error correction with local Markovian dynamics.
A complementary direction is to relax the choice of elementary errors. The transference bounds above rely on the \(\mathbb Z_q\) Pauli algebra of the single-qudit shift errors, and Ref.~\cite{Guo2026quantumcodesleemetric} gives an explicit example with a different local error basis for which this restriction no longer applies and the code distance grows with \(q\). This suggests asking which physically motivated local error models can provide similar protection in two dimensions. More broadly, it would be useful to identify multilevel and bosonic platforms in which the required local stabilizer interactions and relevant local error processes arise naturally.

\begin{acknowledgments}
We thank Andrew Lucas for useful discussions.
This work is supported by the U.S. Department of Energy, Office of Science, Office of Nuclear
Physics, InQubator for Quantum Simulation (IQuS)\footnote{\url{https://iqus.uw.edu/}} under Award Number DOE (NP) Award
DE-SC0020970 via the program on Quantum Horizons\footnote{\url{https://science.osti.gov/np/Research/Quantum-Information-Science}}: QIS Research and Innovation for
Nuclear Science.
This manuscript was prepared with the assistance of ChatGPT, developed by OpenAI, and Claude, developed by Anthropic. These tools were used to improve clarity, generate figures, and assist in generating code for numerical calculations, in accordance with the author’s instructions. All
content, claims, and conclusions have been reviewed and
verified by the author to ensure accuracy and originality.

\end{acknowledgments}

\bibliography{thebib}

\clearpage
\onecolumngrid

\appendix
\setcounter{secnumdepth}{2}
\renewcommand{\thesubsection}{\thesection.\arabic{subsection}}
\makeatletter
\renewcommand{\p@subsection}{}
\makeatother

\section{Proof of $Z$-error sector} \label{app:Z}

In this appendix, we prove an exponential lower bound on the \(Z\)-sector memory time for the discrete-time single-site Metropolis dynamics initialized at an arbitrary codeword.
Our proof proceeds by constructing a bottleneck around each codeword that separates it from all other codewords. We then show that, under the single-site Metropolis dynamics, the probability of reaching this bottleneck within an exponentially long time is small for sufficiently large \(q\) and $L$.

We first lift the $Z$-error configurations to the integral chain complex. For any integer error class $[z]\in C_1(\mathbb{Z})/\im\partial_2$, the integer syndrome energy is defined by
\begin{equation}
\widetilde E_Z([z])
=
\|\partial_1z\|_1.
\end{equation}
On the integer error classes, we consider the same single-site \(\pm1\) Metropolis dynamics with energy \(\widetilde E_Z\).

Let $\chi_L\in C^1(\mathbb{Z})$ be supported on the $z$-oriented edges and assign
coefficient $2^{k-1}$ to every $z$-oriented edge in layer $k$,
\begin{equation}
(\chi_L)_{i,j,k}=2^{k-1},
\qquad
1\leq i,j,k\leq L.
\end{equation}
The purpose of $\chi_L$ is to measure displacement along the logical direction generated by \(\ell_Z\): since $\partial_2^T\chi_L=0$, for any error class $[z]$, we can define the quantity\begin{equation}
T([z])=\langle\chi_L,z\rangle,
\end{equation}
which is invariant when adding any \(Z\) stabilizer to \(z\). If we apply one logical \(Z\), \(\ell_Z\), \(T\) is changed by
\begin{equation}
P_L:=T(\ell_Z)
=
\sum_{k=1}^{L}4^{k-1}
=
\frac{4^L-1}{3}.
\end{equation}
On the other hand, a single-site \(\pm1\) update changes \(T\) by at most the largest coefficient of \(\chi_L\),
\begin{equation}
\Delta_L:=2^{L-1}.
\end{equation}

Let
\begin{equation}
S_q:=\frac{q-1}{2}.
\end{equation}
Define an integer ``cylinder''
\begin{equation}
\widetilde{\mathcal C}
=
\left\{
[z]:
-\frac{P_L}{2}
\leq
T([z])
<
\frac{P_L}{2},
\quad
|(\partial_1z)_a|
\leq
S_q
\text{ for every syndrome coordinate }a
\right\}.
\end{equation}
We call \(\widetilde{\mathcal C}\) a cylinder because \(T\) serves as the longitudinal coordinate, with the two bounds on \(T\) defining the caps of the boundary, while the coordinatewise syndrome bounds define the side.

We now estimate the time needed to reach the boundary of the cylinder.
Starting from $T=0$, the two caps cannot be reached before
\begin{equation} \label{eq:cap-bound}
\tau_{\mathrm{cap}}
:=
\left\lceil
\frac{P_L}{2\Delta_L}
\right\rceil
\geq
2^{L-2},
\qquad
L\geq2.
\end{equation}
As for the side, because one update can change an individual syndrome coordinate by two, define
the side shell
\begin{equation}
\widetilde{\mathcal C}_{\mathrm{side}}
=
\left\{
[z]\in\widetilde{\mathcal C}:
S_q-1
\leq
|(\partial_1z)_a|
\leq
S_q
\text{ for at least one }a
\right\}.
\end{equation}
If an integer syndrome coordinate leaves $[-S_q,S_q]$ in one update, the
configuration immediately before that update must belong to
$\widetilde{\mathcal C}_{\mathrm{side}}$. Next, we estimate the time needed to reach $\widetilde{\mathcal C}_{\mathrm{side}}$. We first need the following lemma.

\begin{lem}[One state per integer syndrome] \label{lem:onestatepersyndrome}
If $[z],[z']\in\widetilde{\mathcal C}$ have the same integer syndrome, then
\begin{equation}
[z]=[z']
\qquad
\text{in }
C_1(\mathbb{Z})/\im\partial_2.
\end{equation}
\end{lem}

\begin{proof}
Equal syndrome implies
\begin{equation}
z-z'\in\ker\partial_1.
\end{equation}
Since $H_1(C_\bullet;\mathbb{Z})\cong\mathbb{Z}$ is generated by $\ell_Z$,
there exist $M\in\mathbb{Z}$ and $a\in C_2(\mathbb{Z})$ such that
\begin{equation}
z-z'
=
M\ell_Z+\partial_2a.
\end{equation}
Applying $T$ gives
\begin{equation}
T([z])-T([z'])
=
MP_L.
\end{equation}
Both values lie in the same half-open interval of length $P_L$, so $M=0$.
Hence $z-z'\in\im\partial_2$.
\end{proof}

\begin{lem}[Side-hitting bound] \label{lem:Side-hitting-bound}
Let
\begin{equation}
m_L=\dim C_0=L^2(L-1),
\end{equation}
and let
\begin{equation}
\widetilde{W}_Z(x)=e^{-\beta\widetilde{E}_Z(x)}
\end{equation}
denote the Boltzmann weight. For every integer $\tau<\tau_{\mathrm{cap}}$ satisfying
\begin{equation} \label{eq:side-time-condition}
\tau\leq\frac{1}{4\widetilde{W}_Z(\widetilde{\mathcal C}_{\mathrm{side}})},
\end{equation}
the probability that $\widetilde{\mathcal C}_{\mathrm{side}}$ is visited during the first $\tau$ steps is at most $1/4$. In particular, it is sufficient that
\begin{equation} \label{eq:side-time-explicit}
\tau\leq
\frac{
\exp\left[
\beta(S_q-1)
-(m_L-1)\log\left(\frac{e^\beta+1}{e^\beta-1}\right)
\right]
}{
8m_L(1+e^{-\beta})
}.
\end{equation}

\end{lem}

\begin{proof}
By Lemma~\ref{lem:onestatepersyndrome}, each integer syndrome labels at most one state in $\widetilde{\mathcal C}$.
Fix one syndrome coordinate $a$. The total Boltzmann weight of the shell values $\pm(S_q-1)$ and $\pm S_q$
is bounded by
\begin{equation}
2e^{-\beta(S_q-1)}+2e^{-\beta S_q}
=
2(1+e^{-\beta})e^{-\beta(S_q-1)}.
\label{eq:side-shell-weight}
\end{equation}

For every other syndrome coordinate, enlarging the allowed interval to all of $\mathbb{Z}$, we get
\begin{equation}
\sum_{r\in\mathbb{Z}}e^{-\beta|r|}
=
1+2\sum_{r=1}^{\infty}e^{-\beta r}
=
\frac{e^\beta+1}{e^\beta-1}.
\label{eq:side-geometric-sum}
\end{equation}

Multiplying these factors and taking a union bound over the $m_L$ possible shell coordinates gives
\begin{equation}
\widetilde{W}_Z(\widetilde{\mathcal C}_{\mathrm{side}})
:=
\sum_{[z]\in\widetilde{\mathcal C}_{\mathrm{side}}}
e^{-\beta\widetilde{E}_Z([z])}
\leq
2m_L(1+e^{-\beta})e^{-\beta(S_q-1)}
\left(
\frac{e^\beta+1}{e^\beta-1}
\right)^{m_L-1}.
\label{eq:side-weight-bound}
\end{equation}

Reversibility of the dynamics and $\widetilde{W}_Z(0)=1$ imply, for every $t\geq0$,
\begin{equation}
P^t(0,x)
=
\widetilde{W}_Z(x)P^t(x,0)
\leq
\widetilde{W}_Z(x).
\label{eq:side-reversibility}
\end{equation}

Consequently, for every integer $\tau<\tau_{\mathrm{cap}}$,
\begin{equation}
\Pr_0\!\left[
\widetilde{\mathcal C}_{\mathrm{side}}
\text{ is visited during the first $\tau$ steps}
\right]
\leq
\tau\widetilde{W}_Z(\widetilde{\mathcal C}_{\mathrm{side}}).
\label{eq:side-hitting-probability}
\end{equation}
Thus \eqref{eq:side-time-condition} makes this probability at most $1/4$, and substituting
\eqref{eq:side-weight-bound} gives \eqref{eq:side-time-explicit}.
\end{proof}

Combining the cap bound (\ref{eq:cap-bound}) with Lemma \ref{lem:Side-hitting-bound}, the \(1/4\) exit time from the integer cylinder satisfies
\begin{equation}
\widetilde t_Z(1/4)
\ge
\min\left\{
2^{L-2},
\,
\frac{
\exp\!\left[
\beta(S_q-1)
-
(m_L-1)\log\!\left(
\frac{e^\beta+1}{e^\beta-1}
\right)
\right]
}{
8m_L(1+e^{-\beta})
}
\right\}.
\end{equation}

We now reduce the cylinder modulo $q$. Let \(\pi_q\) denote reduction modulo \(q\), followed by identifying configurations that differ by \(Z\) stabilizers.

\begin{lem}[Injectivity of \(\pi_q\)] \label{lem:injectivity}
The map $\pi_q$ is injective on $\widetilde{\mathcal C}$, and for every
$[z]\in\widetilde{\mathcal C}$,
\begin{equation}
E_Z(\pi_q[z])
=
\widetilde E_Z([z]).
\end{equation}
Moreover, consider the modulo-$q$ single-site Metropolis chain within $\mathcal C
=
\pi_q(\widetilde{\mathcal C})$. Every single-site update whose integer initial state lies in \(\widetilde{\mathcal C}\setminus\widetilde{\mathcal C}_{\mathrm{side}}\) and whose final state remains in \(\widetilde{\mathcal C}\) is mapped by \(\pi_q\) to the corresponding modulo-\(q\) single-site update with the same energy change and the same Metropolis acceptance probability.
\end{lem}

\begin{proof}
Suppose $[z],[z']\in\widetilde{\mathcal C}$ have the same reduction modulo
$q$. Their reduced syndromes are equal. Since $|(\partial_1z)_a|, |(\partial_1z')_a| \leq S_q$ for every syndrome coordinate $a$, we have
\begin{equation}
\partial_1z=\partial_1z',
\end{equation}
and Lemma \ref{lem:onestatepersyndrome} gives $[z]=[z']$, so $\pi_q$ is injective on $\widetilde{\mathcal C}$. Similarly,
\begin{equation}
|(\partial_1z)_a|_L
=
|(\partial_1z)_a|
\end{equation}
for every syndrome coordinate $a$, and therefore
$E_Z(\pi_q[z])=\widetilde E_Z([z])$.

Every state in \(\mathcal C\) has a unique lift in \(\widetilde{\mathcal C}\). Suppose the integer initial state lies in \(\widetilde{\mathcal C}\setminus\widetilde{\mathcal C}_{\mathrm{side}}\). Then every syndrome coordinate has absolute value at most \(S_q-2\). Since a single-site update changes each syndrome coordinate by at most \(2\), if the updated integer state remains in \(\widetilde{\mathcal C}\), its syndrome also lies in \([-S_q,S_q]\). The initial and final syndromes, and hence the initial and final energies, are therefore identical before and after reduction modulo \(q\). Thus the corresponding updates have the same energy change and the same Metropolis acceptance probability.
\end{proof}

Since \(\pi_q\) is injective on \(\widetilde{\mathcal C}\) and each integer syndrome labels at most one state in \(\widetilde{\mathcal C}\), the $q$ logical translates
\begin{equation}
\mathcal C+M\ell_Z,
\qquad
M\in\mathbb{Z}_q,
\end{equation}
are pairwise disjoint.
Thus the logical translates of \(\mathcal C\) correspond to distinct logical sectors. We can then choose a decoder that assigns logical label \(M\) to every state in \(\mathcal C+M\ell_Z\), with the label outside these cylinders chosen arbitrarily. Since the cylinders are pairwise disjoint, this assignment is well defined. Therefore, a trajectory starting from the zero error cannot change its logical sector before leaving \(\mathcal C\). By Lemma \ref{lem:injectivity}, before this exit the modulo-\(q\) dynamics is identical to the integer dynamics in \(\widetilde{\mathcal C}\). Therefore the \(Z\)-sector memory time satisfies
\begin{equation}
t_{\mathrm{mem}}^Z(1/4)
\geq
\widetilde t_Z(1/4)
\geq
\min\left\{
2^{L-2},
\,
\frac{
\exp\!\left[
\beta(S_q-1)
-
(m_L-1)\log\!\left(
\frac{e^\beta+1}{e^\beta-1}
\right)
\right]
}{
8m_L(1+e^{-\beta})
}
\right\}.
\end{equation}

\begin{cor}[Temperature-independent choice of \(q\)]
Let $q=q(L)$ be an odd sequence chosen independently of temperature and
suppose
\begin{equation}
q=\omega(L^3).
\end{equation}
Then, for every fixed $\beta>0$, there exists $L_0(\beta)$ such that for all
$L\geq L_0(\beta)$,
\begin{equation}
t_{\mathrm{mem}}^Z(1/4)
\geq
2^{L-2}
=
e^{\Omega(L)}.
\end{equation}
For a fixed \(\beta>0\), the condition \(q=\omega(L^3)\) can be weakened to \(q=cL^3\), with any fixed \(c\) satisfying
\begin{equation} \label{eq:c-condition}
c>
\frac{2}{\beta}
\log\!\left(
\frac{e^\beta+1}{e^\beta-1}
\right).
\end{equation}
\end{cor}

\begin{proof}
Using $S_q=(q-1)/2$ and $m_L=L^2(L-1)$, the logarithm of the second term in the memory-time
lower bound is
\begin{equation}
\frac{\beta(q-3)}{2}
-
\left(L^2(L-1)-1\right)
\log\left(\frac{e^\beta+1}{e^\beta-1}\right)
-
\log\left[8L^2(L-1)(1+e^{-\beta})\right].
\label{eq:memory-time-log}
\end{equation}
For every fixed $\beta>0$, the first term in \eqref{eq:memory-time-log} is $\omega(L^3)$, while the second term is $O(L^3)$ and the last term is
$O(\log L)$. Therefore, the whole expression in \eqref{eq:memory-time-log} is $\omega(L^3)$.
Hence, for sufficiently large $L$, the second term in the memory-time lower bound is larger than
$2^{L-2}$. The minimum is therefore set by the cap bound, giving
\begin{equation}
t_{\mathrm{mem}}^Z(1/4)
\geq
2^{L-2}
=
e^{\Omega(L)}.
\label{eq:z-memory-final}
\end{equation}
Note that the result holds for every fixed
finite temperature, with the required system size $L_0(\beta)$ allowed to depend on the temperature.

For a fixed $\beta>0$, if $q=cL^3$, the leading $L^3$ contribution to \eqref{eq:memory-time-log}
is
\begin{equation}
\left[
\frac{\beta c}{2}
-
\log\left(\frac{e^\beta+1}{e^\beta-1}\right)
\right]L^3.
\label{eq:memory-time-leading}
\end{equation}
Hence any $c$ satisfying \eqref{eq:c-condition} makes the second term in the memory-time lower bound
exponentially larger than $2^{L-2}$ for sufficiently large $L$, giving the same bound
$t_{\mathrm{mem}}^Z(1/4)\geq 2^{L-2}$.
\end{proof}

\subsection{Heuristic analysis of the \(Z\)-sector entropic bound} \label{app:Z-q}

The rigorous bound above uses a very loose estimate of the entropic contribution. Specifically, \eqref{eq:side-weight-bound} introduces a factor that grows exponentially with $m_L=\Theta(L^3)$ by summing independently over all remaining syndrome coordinates. This step is the main source of the stringent requirement on $q$ in the analytical result.

To illustrate why this factor may substantially overestimate the entropic contribution, suppose for comparison that the $m_L$ syndrome coordinates fluctuate independently. Let $Z_1=\frac{e^\beta+1}{e^\beta-1}$. The total partition function would then be
\begin{equation}
Z_{\mathrm{ind}}=Z_1^{m_L}.
\label{eq:independent-partition-function}
\end{equation}
If one syndrome coordinate lies in the side shell, the remaining $m_L-1$ coordinates contribute the factor $Z_1^{m_L-1}$. Using \eqref{eq:side-weight-bound}, the corresponding normalized probability is therefore bounded by
\begin{equation}
P_{\mathrm{side}}^{\mathrm{ind}}
\lesssim
\frac{
2m_L(1+e^{-\beta})e^{-\beta(S_q-1)}Z_1^{m_L-1}
}{
Z_1^{m_L}
}
=
\frac{2m_L(1+e^{-\beta})}{Z_1}
e^{-\beta(S_q-1)}.
\label{eq:independent-side-probability}
\end{equation}
Thus the extensive factor $Z_1^{m_L-1}$ is removed by normalization. What remains is the factor $m_L$, which counts the possible syndrome coordinates that can reach the side shell, together with
\begin{equation}
e^{-\beta(S_q-1)}
=
e^{-\beta q/2+O(1)}.
\label{eq:independent-side-suppression}
\end{equation}
Under this independence assumption, the corresponding memory time would therefore scale as
\begin{equation}
t_{\mathrm{mem}}^Z
\sim
\operatorname{poly}(L)^{-1}e^{\beta q/2},
\label{eq:independent-memory-time}
\end{equation}
up to subexponential corrections.

The syndrome coordinates in the actual code are constrained and are not independent, so Eqs.~\eqref{eq:independent-partition-function}--\eqref{eq:independent-memory-time} are only a heuristic comparison and do not provide a rigorous bound for the original model. Numerically, at $\beta=1$, we find $t_{\mathrm{mem}}^Z\propto e^{0.477q}$, with an exponent close to $1/2$. This is consistent with the extensive entropic factor in \eqref{eq:side-weight-bound} substantially overestimating the entropic contribution, and suggests that a tighter analytical bound may allow a weaker scaling of $q$ in the rigorous result.

\section{Proof of $X$-error sector} \label{app:X}

In this appendix, we prove an exponential lower bound on the \(X\)-sector memory time starting from the thermal encoding of an arbitrary codeword.
The logical \(X\) operators in our code have the same form as in the ordinary 3D \(\mathbb Z_q\) surface code, and the stability proof for the \(X\)-error sector follows the same basic strategy as the thermal-stability analysis of the 4D toric code~\cite{Dennis2002topological,Alicki2010thermal}. We show that any sequence of local errors that produces a logical \(X\) error must pass through a configuration whose \(Z\) syndrome contains a connected cluster of size \(\Omega(L)\). We then show that, at sufficiently low temperature, the total Gibbs weight of configurations containing such a cluster is exponentially small in \(L\), which in turn yields an exponentially long memory time in the \(X\)-error sector.

\begin{lem} [Support-preserving chain isomorphism] \label{lem:deformation}
By replacing $K_L$ with $-R_L$ in (\ref{eq:chaincomplex1}), we obtain the ordinary 3D surface code over $\mathbb Z_q$. Let $H_Z^{(0)}$ be the $Z$-check matrix of the ordinary 3D surface code. For every odd $q$, the chain complexes of weighted and ordinary 3D surface codes are related by an invertible diagonal chain isomorphism. Moreover, there is a correction map $\mathsf{Rec}$ and a constant $\alpha>0$, independent of \(L\) and \(q\), such that if a single-site \(X^{\pm1}\) transition \(x\to x'\) has no connected component of either \(H_Zx\) or \(H_Zx'\) of size at least \(\alpha L\), then
\begin{equation}
x'-\mathsf{Rec}(H_Zx')-x+\mathsf{Rec}(H_Zx)
\in \operatorname{im}H_X^T .
\end{equation}
Thus such a transition cannot change the recovered logical \(X\) class.
    
\end{lem}

\begin{proof}
Since $q$ is odd, every power of $2$ is a unit in $\mathbb{Z}_q$. Define
\begin{subequations}
\begin{equation}
V_1=\operatorname{diag}\left(1,2,2^2,\ldots,2^{L-1}\right),
\end{equation}
\begin{equation}
V_0=\operatorname{diag}\left(2,2^2,\ldots,2^{L-1}\right).
\end{equation}
\end{subequations}
A direct calculation gives
\begin{equation}
K_L=V_0(-R_L)V_1^{-1}\pmod q.
\end{equation}
Tensoring these maps with the identity maps in the $x$ and $y$ directions induces diagonal isomorphisms $D_{\mathrm{err}}$ and
$D_{\mathrm{syn}}$ on the $X$-error and $Z$-syndrome spaces, respectively, with
\begin{equation}
H_Z=D_{\mathrm{syn}}H_Z^{(0)}D_{\mathrm{err}}^{-1}.
\end{equation}
The induced cochain isomorphism maps ordinary $X$ stabilizers to weighted $X$ stabilizers and therefore preserves
the logical $X$ sectors. Because all diagonal entries are units, $D_{\mathrm{err}}$ and $D_{\mathrm{syn}}$ do not change which coordinates are zero. Hence they preserve error support and syndrome support exactly.

It therefore remains to construct a local recovery for the ordinary 3D $\mathbb{Z}_q$ surface code.
In the dual lattice of the ordinary $\mathbb{Z}_q$ surface code, a physical qudit (edge), a $Z$-check plaquette, and a cube on the primal lattice correspond to a dual plaquette, dual edge, and dual vertex, respectively. Let the boundary maps of the dual chain complex, with the ordinary surface-code boundary conditions understood, be
\begin{equation}
\widetilde{C}_3
\xrightarrow{\widetilde{\partial}_3}
\widetilde{C}_2
\xrightarrow{\widetilde{\partial}_2}
\widetilde{C}_1
\xrightarrow{\widetilde{\partial}_1}
\widetilde{C}_0,
\label{eq:dual-chain-complex}
\end{equation}
and choose orientations such that $\widetilde{\partial}_2$ corresponds to $H_Z^{(0)}$.  Two syndrome dual edges with nonzero coefficients are called adjacent if they share a dual vertex, and the connected components under this adjacency relation are called connected syndrome clusters.

Let $\gamma$ be the dual $1$-chain corresponding to the full syndrome. Since it is the boundary of a dual $2$-chain, we have $\widetilde{\partial}_1\gamma=0$ and $\widetilde{\partial}_1\gamma_C=0$ for each connected syndrome cluster $\gamma_C$. Also, because the full dual complex has trivial first homology,
\begin{equation}
H_1(\widetilde C_\bullet;\mathbb Z_q)=0,
\end{equation}
every \(\gamma_C\in\ker\widetilde{\partial}_1\) lies in \(\operatorname{im}\widetilde{\partial}_2\). Hence every connected syndrome cluster is itself the boundary of a dual \(2\)-chain and therefore is a valid syndrome.

Let $m=|C|$. Since the dual lattice has bounded local degree, there exists a constant $c>0$, independent of $L$ and $q$, such that $C$ is contained in a cubical neighborhood $U_C$ satisfying $\operatorname{diam}(U_C)\leq cm$.
Choose $\alpha>0$ sufficiently small such that whenever
\begin{equation}
m<\alpha L,
\label{eq:dual-small-cluster}
\end{equation}
the region $U_C$ does not connect any pair of opposite physical boundaries, and hence cannot support any nontrivial logical \(X\) operator.
For such a region $U_C$, the local dual complex with the chosen ordinary surface-code boundary conditions has no nontrivial homology in degrees one and two:
\begin{equation}
H_1(U_C;\mathbb{Z}_q)
=
H_2(U_C;\mathbb{Z}_q)
=
0.
\label{eq:local-dual-homology}
\end{equation}
If $U_C$ intersects a physical boundary on which syndrome endpoints are allowed, \eqref{eq:local-dual-homology} is understood as the corresponding relative homology statement. By $\widetilde{\partial}_1\gamma_C=0$ and~\eqref{eq:local-dual-homology}, there exists a dual $2$-chain $\Sigma_C$, supported entirely in $U_C$, such that
\begin{equation}
\widetilde{\partial}_2\Sigma_C=\gamma_C.
\label{eq:local-dual-filling}
\end{equation}
Mapping \(\Sigma_C\) back to the primal lattice gives an ordinary-code \(X\)-error configuration \(y_C^{(0)}\), which is then mapped by the support-preserving chain isomorphism to a weighted-code \(X\) recovery $y_C$. Hence, for the $X$-error sector of the weighted 3D surface code, every connected syndrome cluster of size smaller than \(\alpha L\) admits a local \(X\) recovery supported entirely in \(U_C\).

For a valid syndrome $\gamma$, if all connected components $\gamma_C$ have size smaller than $\alpha L$, we can define
\begin{equation}
\mathsf{Rec}(\gamma)=\sum_C y_C
\label{eq:recovery-map}
\end{equation}
as the correction map.
For a valid syndrome containing a connected component of size at least $\alpha L$, fix any recovery; the locality of these recoveries will not be needed below.

Now consider a single-site $X^{\pm1}$ transition $x\to x'$, and suppose that all connected components of both $H_Zx$ and $H_Zx'$ have size smaller than $\alpha L$. This update can only change finitely many connected components. After decreasing $\alpha$ by a fixed constant factor if necessary, their local recoveries, together with the support of $x'-x$, can therefore all be contained in a region $U$ that does not connect any pair of opposite physical boundaries. Since any nontrivial logical $X$ operator connects opposite boundaries,
\begin{equation}
x'-\mathsf{Rec}(H_Zx')
-x+\mathsf{Rec}(H_Zx)
\in\operatorname{im}H_X^T.
\label{eq:recovered-logical-class-invariant}
\end{equation}
This completes the proof.

\end{proof}

After decreasing \(\alpha\) by a fixed constant factor if necessary, Lemma~\ref{lem:deformation} implies that any syndrome for which a single-site \(X^{\pm1}\) update can change the recovered logical \(X\) class contains a connected component of size at least \(\alpha L\). We next bound the Gibbs probability of such syndromes by a Peierls-type argument~\cite{Peierls1936ising,Griffiths1964peierls}.
\begin{lem}[Peierls estimate]
\label{lem:uniform-lee-activity-peierls}
There exists a finite \(\beta_X\) such that for every odd \(q\) and every fixed \(\beta>\beta_X\),
the Gibbs probability $P_{\mathrm{crit}}^X$ that the $Z$ syndrome contains a connected component of size at least $\alpha L$ satisfies
\begin{equation}
P_{\mathrm{crit}}^X
\leq
\operatorname{poly}(L)e^{-\kappa_X L}
\label{eq:x-critical-probability}
\end{equation}
for some $\kappa_X>0$ independent of $L$ and $q$.
\end{lem}
\begin{proof}
Since $E_X(x):=|H_Zx|_L$,
the Gibbs weight of a valid syndrome $s$ is proportional to $e^{-\beta|s|_L}$.
By the support-preserving chain isomorphism established above, every connected component of an allowed weighted syndrome is itself a valid syndrome. Fix such a component $C$, including its nonzero $\mathbb{Z}_q$ values. Removing $C$ from any valid syndrome in which it occurs gives an injective map into the set of valid syndromes. Since $C$ has disjoint support from the remaining syndrome, $|s|_L=|C|_L+|s-C|_L$,
so
\begin{equation}
\Pr[C\text{ occurs}]
\leq
e^{-\beta|C|_L}.
\label{eq:component-occurrence-bound}
\end{equation}

Let $\Delta$ be the maximum degree of the syndrome adjacency graph. After fixing an ordering of the adjacent checks, each connected support of size $m$ containing a fixed $Z$ check can be assigned a depth-first traversal of a spanning tree of length at most $2(m-1)$. Hence the number of such supports is at most $\Delta^{2m}$.
Since\begin{equation}
\sum_{a\in\mathbb{Z}_q\setminus\{0\}} e^{-\beta|a|_L}
=
2\sum_{r=1}^{(q-1)/2}e^{-\beta r}
\leq
\frac{2}{e^\beta-1},
\label{eq:uniform-lee-activity}
\end{equation}
for any fixed support of size $m$, summing \eqref{eq:component-occurrence-bound} over all nonzero syndrome values on that support and using \eqref{eq:uniform-lee-activity} gives a total weight at most
\begin{equation}
\left(
\frac{2}{e^\beta-1}
\right)^m.
\label{eq:fixed-support-total-weight}
\end{equation}
The total number of $Z$ checks is $N_Z=(L-1)^3+2L^2(L-1)$.
A union bound therefore gives
\begin{equation}
P_{\mathrm{crit}}^X
\leq
N_Z
\sum_{m\geq\alpha L}
\left(
\frac{2\Delta^2}{e^\beta-1}
\right)^m.
\label{eq:x-critical-union-bound}
\end{equation}
Choose $\beta_X$ so that
\begin{equation}
\frac{2\Delta^2}{e^{\beta_X}-1}
<
1.
\label{eq:beta-x-condition}
\end{equation}
For every fixed $\beta>\beta_X$, the geometric series in \eqref{eq:x-critical-union-bound} is exponentially small in $L$, which proves \eqref{eq:x-critical-probability} with constants independent of $q$.
\end{proof}

Let \(A_m\), \(m\in\mathbb Z_q\), denote the configurations with recovered logical \(X\) class \(m\), and let \(\pi\) denote the Gibbs distribution. By logical-translation symmetry, \(\pi(A_m)=1/q\), and the stationary one-step probability of leaving \(A_m\) is the same for every \(m\). By Lemma~\ref{lem:deformation}, any transition that changes the recovered logical \(X\) class has at least one endpoint whose syndrome contains a connected component of size at least \(\alpha L\). Therefore, by stationarity,
\begin{equation}
\sum_{m\in\mathbb{Z}_q}
\Pr_{\pi}
\!\left(
X_0\in A_m,\,
X_1\notin A_m
\right)
\leq
2P_{\mathrm{crit}}^X.
\label{eq:x-sector-total-exit-flow}
\end{equation}
Consequently,
\begin{equation}
\frac{
\Pr_{\pi}
\!\left(
X_0\in A_0,\,
X_1\notin A_0
\right)
}{
\pi(A_0)
}
\leq
2P_{\mathrm{crit}}^X.
\label{eq:x-sector-normalized-exit-flow}
\end{equation}
For the thermal encoding of logical class $0$, the initial distribution is $\pi(\,\cdot\,|A_0)$. At every time $s\geq0$, stationarity implies that the distribution of $X_s$ is pointwise bounded by $\pi/\pi(A_0)$. Hence \eqref{eq:x-sector-normalized-exit-flow} gives
\begin{equation}
\Pr_{\pi(\,\cdot\,|A_0)}
\!\left(
X_s\in A_0,\,
X_{s+1}\notin A_0
\right)
\leq
2P_{\mathrm{crit}}^X.
\label{eq:x-sector-one-step-exit}
\end{equation}
A union bound therefore yields
\begin{equation}
\Pr_{\pi(\,\cdot\,|A_0)}
\!\left(
\text{the trajectory leaves }A_0
\text{ within the first }t\text{ steps}
\right)
\leq
2tP_{\mathrm{crit}}^X.
\label{eq:x-sector-exit-union-bound}
\end{equation}
Using Lemma~\ref{lem:uniform-lee-activity-peierls}, for any fixed $0<c_X<\kappa_X$,
\begin{equation}
2e^{c_XL}P_{\mathrm{crit}}^X
\leq
2\,\operatorname{poly}(L)
e^{-(\kappa_X-c_X)L}
\longrightarrow
0.
\label{eq:x-sector-exit-exponential}
\end{equation}
Thus, for all sufficiently large $L$, the probability of leaving $A_0$ within $e^{c_XL}$ steps is smaller than $1/4$. Since decoding failure at time $t$ requires the trajectory to have left $A_0$ by time $t$,
\begin{equation}
t_{\mathrm{mem}}^X(1/4)
\geq
e^{c_XL},
\label{eq:x-sector-memory-time}
\end{equation}
with $c_X>0$ independent of $L$ and $q$.

\end{document}